\documentclass[11pt]{article}
\usepackage{fullpage}
\usepackage{float}
\usepackage{epsfig}
\usepackage{amsmath}
\usepackage{amssymb}
\usepackage{amstext}
\usepackage{amsthm}
\usepackage{xspace}
\usepackage{latexsym}
\usepackage{verbatim}
\usepackage{multirow}
\usepackage{ifthen}
\usepackage{url}
\usepackage{tcolorbox}
\usepackage{cancel}
\usepackage{soul} 
\usepackage{subcaption}
\usepackage[linesnumbered,ruled]{algorithm2e}
\usepackage{mdframed}
\usepackage{dsfont}
\usepackage[breaklinks=true]{hyperref}

\RestyleAlgo{boxruled}
\DeclareMathAlphabet{\mathitbf}{OML}{cmm}{b}{it}
\newtheorem{theorem}{Theorem}[section]
\newtheorem{definition}{Definition}

\newtheorem{lemma}[theorem]{Lemma}

\newtheorem{proposition}[theorem]{Proposition}

\newtheorem{example}{Example}

\renewcommand{\comment}[1]{}

\definecolor{gray230}{RGB}{240,240,240}

\SetCommentSty{mycommfont}
\newcommand{\prob}[2][]{\text{\bf Pr}\ifthenelse{\not\equal{}{#1}}{_{#1}}{}\!\left[#2\right]}
\newcommand{\expect}[2][]{\text{\bf E}\ifthenelse{\not\equal{}{#1}}{_{#1}}{}\!\left[#2\right]}

\newcommand{\agind}[1][i]{_{#1}}

\newcommand{\inversed}[1]{#1^{-1}}
\newcommand{\ironed}{\bar}

\newcommand{\differentiated}[1]{#1'}
\newcommand{\fortype}{\tilde}
\newcommand{\estimated}{\hat}
\newcommand{\sampled}{\hat}

\newcommand{\noaccents}[1]{#1}
\newcommand{\composed}[3]{#1{#2{#3}}}

\newcommand{\newindexedvar}[4][\noaccents]{%
\expandafter\newcommand\expandafter{\csname #2\endcsname}{#1{#4}}%
\expandafter\newcommand\expandafter{\csname #2s\endcsname}{#1{\boldsymbol{#4}}}%
\expandafter\newcommand\expandafter{\csname #2sm#3\endcsname}[1][#3]{#1{\boldsymbol{#4}}_{-##1}}%
\expandafter\newcommand\expandafter{\csname #2#3\endcsname}[1][#3]{#1{#4}\agind[##1]}%
\expandafter\newcommand\expandafter{\csname #2#3th\endcsname}[1][#3]{#1{#4}_{(##1)}}%
}

\newindexedvar{wal}{k}{w}
\newindexedvar[\differentiated]{margwal}{k}{\wal}
\newindexedvar{cumwal}{k}{W}

\newindexedvar[\ironed]{iwal}{k}{\wal}
\newindexedvar[\ironed]{cumiwal}{k}{\cumwal}
\newindexedvar[\composed{\differentiated}{\ironed}]{margiwal}{k}{\wal}

\newindexedvar{yal}{k}{y}
\newindexedvar{cumyal}{k}{Y}
\newindexedvar[\ironed]{iyal}{k}{\yal}
\newindexedvar[\composed{\differentiated}{\ironed}]{margiyal}{k}{\yal}

\newindexedvar{murev}{k}{P}
\newindexedvar[\estimated]{emurev}{k}{\murev}
\newindexedvar[\differentiated]{mumarg}{k}{\murev}

\newindexedvar[\ironed]{imurev}{k}{\murev}
\newindexedvar[\composed{\differentiated}{\ironed}]{imumarg}{k}{\murev}

\newindexedvar{val}{i}{v}

\newindexedvar{eff}{i}{e}

\newindexedvar{bid}{i}{b}
\newindexedvar[\estimated]{ebid}{i}{\bid}
\newindexedvar[\sampled]{sbid}{i}{\bid}

\newindexedvar{strat}{i}{s}

\newindexedvar{virt}{i}{\phi}
\newindexedvar[\inversed]{virtinv}{i}{\virt}

\newindexedvar[\ironed]{ivirt}{i}{\virt}

\newindexedvar{dist}{i}{F}

\newindexedvar{dens}{i}{f}

\newindexedvar{price}{i}{p}
\newindexedvar[\fortype]{tprice}{i}{p}

\newindexedvar{alloc}{i}{x}

\newindexedvar[\fortype]{talloc}{i}{\alloc}

\newindexedvar{util}{i}{u}
\newindexedvar[\fortype]{tutil}{i}{u}

\begin{document}
\title{A Budget Feasible Peer Graded Mechanism For IoT-Based Crowdsourcing}
\author{Vikash Kumar Singh\thanks{Department of Computer Science \& Engineering, NIT Durgapur, Durgapur, West Bengal, India}~ \footnote{email: vikas.1688@gmail.com}
\and
Sajal Mukhopadhyay\footnotemark[1]
\and
Fatos Xhafa\thanks{Department of Computer Science, Universitat Polit$\grave{e}$cnica de Catalunya, Barcelona, Spain}
\and
Aniruddh Sharma\thanks{Engineer, Qualcomm India Private Ltd., Hyderabad, Andhra Pradesh, India}}

\maketitle

\begin{abstract}
We develop and extend a line of recent works on the design of mechanisms for heterogeneous tasks assignment problem in 'crowdsourcing'. The budgeted market we consider consists of multiple task requesters and multiple IoT devices as task executers; where each task requester is endowed with a single distinct task along with the publicly known budget. Also, each IoT device has valuations as the cost for executing the tasks and quality, which are private. Given such scenario, the objective is to select a subset of IoT devices for each task, such that the total payment made is within the allotted quota of the budget while attaining a threshold quality. For the purpose of determining the unknown quality of the IoT devices we have utilized the concept of \emph{peer grading}. In this paper, we have carefully crafted a \emph{truthful budget feasible} mechanism; namely TUBE-TAP for the problem under investigation that also allows us to have the true information about the \emph{quality} of the IoT devices. The simulations are performed in order to measure the efficacy of our proposed mechanism.     
\end{abstract}



\section{Introduction}
\label{s:intro}
Over the past decades, most of the works in \emph{crowdsourcing} \footnote{https://www.wired.com/2006/06/crowds/}\cite{Howe2006}\cite{Slivkins:2014:ODM:2692359.2692364} mainly circumvent around tackling one of the major challenges of \emph{how to motivate the crowd workers to participate in the system?} One solution that is appreciated a lot in this direction is, to incentivize the task executers. This gave rise to several other open questions: 1) Which task executers to be hired? 2) How the task requester(s) can be aware about the quality of the task executers (or crowd workers)? 3) What amount is to be paid to the task executers for their services, so that they are not dishearten and are motivated to participate in future, in similar type of systems? Answering to the above raised questions, substantial amount of works have been done in these directions \cite{Bhat:2016:TMB:2937029.2937172}\cite{DBLP:journals/corr/abs-1305-6705}\cite{DBLP:conf/hcomp/GoelNS14}\cite{JAIN201844}\cite{Jain:2016:DMM:2936924.2936941}\cite{Luo:2016:IMD:2885506.2837029}\cite{10.1109:MIC.2012.70}. Unlike the works in \cite{Xu:2017:BOA:3091125.3091431}\cite{5466993}, in this paper, we have investigated the set-up somehow close to the set-up discussed in \cite{DBLP:conf/hcomp/GoelNS14}\cite{DBLP:journals/corr/AssadiHJ15} but with additional constraints: 1) the task executers are the IoT devices instead of human agents, and 2) in order to be aware about the quality of IoT devices, we have utilized the technique of peer grading, that is different from the general practice for identifying the quality of the human agents \cite{JAIN201844}\cite{Bhat:2016:TMB:2937029.2937172}. It is to be noted that, till date, in the crowdsourcing literature this tedious work of determining the quality of the crowd workers is mostly done by the platform or in some cases by the task requesters. This leads to an extra burden on the platform or the task requesters. Also, this scenario makes the process of quality determination centralized. In our peer grading approach, we use to distribute the task executed by the subset of IoT devices to their peers (other IoT devices) for grading purpose. Based on the peers report, the quality IoT devices are selected.\\
\indent The detailing of our proposed model is depicted in Figure \ref{fig:model}. In our model, we have multiple task
requesters and multiple IoT devices (as task executers); where each task requester is endowed with a single task and the maximum amount he/she (henceforth he) can pay is termed as \emph{budget} (or capital). Each IoT device has independent private cost(s) for each task that they will charge for executing. It is to be noted that, the participating IoT devices are intelligent and rational. Due to their rational behaviour they will try to strategize the system. By \emph{strategizing} we mean that these devices can manipulate their private information in order to gain. Given this set-up, our goal is to select the subset of IoT devices for each task such that the total payment made to the IoT devices are within the allotted quota of budget for the task while attaining a threshold quality.
 \begin{figure}[H]
\begin{center}
\includegraphics[scale=0.45]{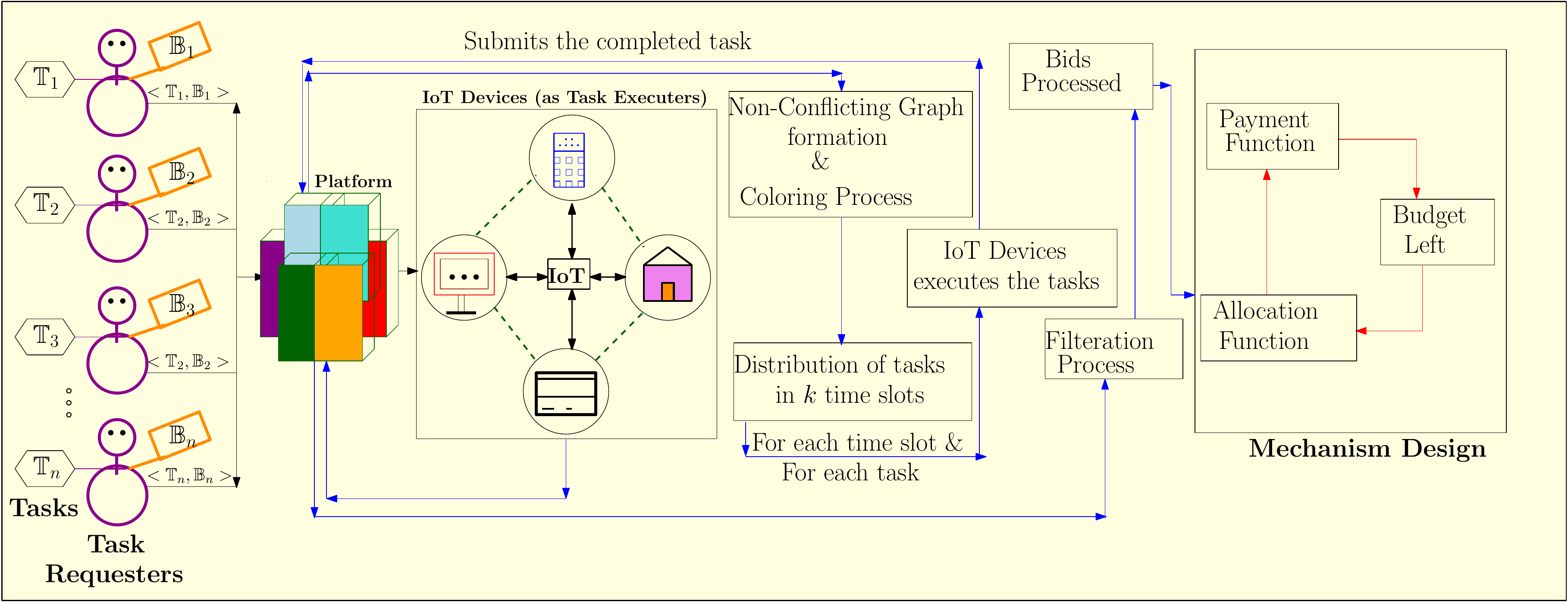}
\caption{A Pictorial representation of proposed model}
\label{fig:model}
\end{center}
\end{figure}
\indent Following the general work flow of the crowdsourcing, firstly, each task requester submits the endowed task and the publicly known budget to the platform. On receiving the tasks and the endowed capital for the respective task from the task requesters, the platform publishes the tasks to the outside world for the execution purpose. Now, each IoT device present on the other side of the market opts for the subset of tasks of their interest for execution and report to the platform along with the amount they will charge for executing each task. Based on their reported interests, the platform assigns the tasks to the IoT devices. Without loss of generality, it is assumed that each IoT device will execute all of its tasks for which it has shown interest and each IoT device executes single task at a time. Now, the immediate question is: \emph{How to preserve the assumptions made for the problem under investigation?} One solution that can be thought of is to place each of the task of an IoT device on which it has shown interest into different time slots (here, time slots could be thought of as \emph{morning}, \emph{afternoon}, and \emph{evening} for a day) that will help in keeping our assumptions alive.\\
\indent \emph{Say, for example an IoT device has shown his interest over 3 tasks. In such case, one task will be scheduled in the morning, another task in the afternoon, and the third task could be scheduled in the evening.}\\
\indent After the distribution of tasks into different time slots, the IoT devices executes the assigned task(s) and submit to the platform as depicted in Figure \ref{fig:model}. Now, the next challenge that comes into the pocket of the \emph{platform} is to determine the quality of the IoT devices. For this purpose, the idea of peer grading \cite{DBLP:journals/corr/AlfaroSP16}\cite{T.roughgarden_201617} is utilized in our set-up. It is to be noted that, in each time slot and for each task, the process of \emph{peer grading} is carried out iteratively. The process iterates until each IoT device is not graded by the peers. At the end of each iteration of the peer grading phase, the IoT device top rated (or graded) by most of the peers, is selected. Finally, the peer grading process returns a set of quality IoT devices for each task. Now, given the set of quality IoT devices for each task, we have to select a subset of IoT devices such that the total payment made are within the allotted quota of budget. As the IoT devices are \emph{strategic} in our setting, so for this reason we have modelled the above discussed set-up using mechanism design. \\
\indent In this paper, we have carefully crafted a \emph{\textbf{\underline{t}}ruthful b\textbf{\underline{u}}dget feasi\textbf{\underline{b}}le m\textbf{\underline{e}}chanism} for the \emph{\textbf{\underline{t}}ask \textbf{\underline{a}}llocation \textbf{\underline{p}}roblem} (TUBE-TAP) motivated by \cite{Singer_2016}\cite{8432319}, that also allow us to have the true information about the \emph{quality} of the IoT devices\footnote{It is to be noted, our proposed system is applicable equally to the system where there are human agents instead of IoT devices in the role of \emph{task executers}.}.
\subsection{Summary of Key Contributions}
The main contributions of this paper are:
\begin{itemize}
\item [-] We have investigated the heterogeneous task assignment problem in IoT based crowdsourcing through the lens of mechanism design.
\item [-] We have developed a \emph{truthful budget feasible} mechanism; namely TUBE-TAP motivated by \cite{Singer_2016}\cite{8432319} for the problem under investigation.
\item [-] We prove that TUBE-TAP satisfies several economic properties such as \emph{truthfulness}, and \emph{budget feasibility}.     
 \item [-] The simulations are done for comparing the TUBE-TAP with a carefully crafted benchmark mechanism.
\end{itemize}
\subsection{Paper Organization}
The remainder of this paper is organized as follows. In section \ref{sec:rw} the prior works explored in the direction of crowdsourcing is discussed. In Section \ref{s:prelim}, we describe our proposed system model in detailed manner. We then present our proposed mechanism namely TUBE-TAP for the problem discussed in section \ref{s:pm}. Further analysis of TUBE-TAP is carried out in section \ref{sec:faba}. In section \ref{sec:ef} the experimental results are presented and discussed. In section \ref{se:cons} the paper is concluded and the future directions are coined.

\section{Related Works}\label{sec:rw}
This section contains a short description of the previous works and developments in this area. The discussion will mainly circumvent around the works regarding \emph{incentive} policies utilized in the past for dragging large number of participants, and \emph{quality} of the executed works supplied by the task executers, in crowdsourcing. In order to get the detailed overview of the field and the current research trends we recommend readers to go through \cite{Howe2006}\cite{6113213}\cite{Slivkins:2014:ODM:2692359.2692364}\cite{Mazlan2018}.\\
\indent In past there have been an extensive body of works discussing about the major challenges in crowdsourcing \cite{Slivkins:2014:ODM:2692359.2692364} and in some cases providing the solution approach \cite{Bhat:2016:TMB:2937029.2937172}\cite{Jain:2016:DMM:2936924.2936941}\cite{Luo:2016:IMD:2885506.2837029}. The two major challenges in crowdsourcing that have dragged the interest of large community are: 1) How to motivate large group of common people to participate in this system as they are \emph{rational}. 2) How to verify that the executed tasks supplied by the agents are upto the mark. Answering to the issue raised in point 1 several schemes are proposed that incentivizes the participating agents in some sense \cite{Luo:2016:IMD:2885506.2837029}\cite{DBLP:conf/hcomp/GoelNS14}\cite{Bhat:2016:TMB:2937029.2937172}\cite{8031314}. In \cite{Reddy:2010:EMP:1864349.1864355} the \emph{fixed price scheme} is proposed in which the platform or in some applications task requesters provide some fixed amount to the crowd workers. The drawback to such approach is that the agents are paid less than the effort supplied. This pricing structure has resolved the issue raised in point 1 to some extent but not completely.\\
\indent In \cite{5466993} reverse auction based incentive scheme called \emph{'RADP'} is proposed for the setting with single task requester having single task that is to be given to the multiple crowd workers on the other side of the market. In this, some pre-defined number of task executers with lowest bid values are selected and paid their revealed bid price. One issue with this solution approach is that the participants those who are giving much high effort may bid high and may not be reaching in the selection zone. In some sense, this pricing model may demotivate the quality agents. To overcome the issue raised in \cite{5466993}, in \cite{Lee2010PMC} a reverse auction based incentive scheme with virtual participant credit (RADP-VPC) is proposed. Here, the idea is, the participant who lost in the current iteration is provided a specific reward for the participation and if the loser participate further then this virtual credit will be subtracted from his original bid value that may lead to the consideration of participants in further auction rounds. One drawback with this strategy is that the participant can set the high bid value as his/her payment. Following works in \cite{Reddy:2010:EMP:1864349.1864355}\cite{Lee2010PMC}, a better auction models were proposed \cite{Zhao2014CCC}\cite{7218673}\cite{6848055}.\\
\indent In \cite{DBLP:conf/hcomp/GoelNS14}, an effort has been made to design a truthful budget feasible mechanism for crowdsourcing in an online environment for the set-up consisting of single task requester endowed with multiple tasks and there are multiple task executers on the other side of the market. The task executers along with the private cost have different skills based on which they show their interest to perform certain subset of tasks. The goal is to select subset of task executers so that the total payment made to the task executers are within budget. In the similar line, the work by \cite{Xu:2017:BOA:3091125.3091431} is carried out where, the set-up consists of multiple tasks with deadlines that are to be executed by the pool of workers that arrive online. Each of the worker has the known set of tasks that he/she can perform and based on that the task is assigned to the workers before its deadline. The goal is to design an online-assignment policy such that the total expected profit is maximized subject to budget and deadline constraint. \\     
\indent However, one of the major set back of the literature covered till now is that the quality of the data supplied or more formally, the quality of the crowd workers are not taken into picture. Some quality adaptive schemes are discussed in \cite{JAIN201844}\cite{DBLP:journals/corr/abs-1305-6705}. In this paper, an effort has been made to design a quality adaptive truthful budget feasible mechanism for one of the scenarios of 'crowdsourcing'. We have utilized the concept of '\emph{peer grading}' for determining the quality of the IoT devices.

\section{System Model and Problem Formulation}\label{s:prelim}

%
%
%
%
%

In this section, we present the formal statement of our problem. We consider \emph{n} task requesters $\mathbb{R} = \{\mathbb{R}_1,\mathbb{R}_2, \ldots, \mathbb{R}_n\}$ each carrying a single distinct task. The set of tasks is represented as $\mathbb{T} = \{\mathbb{T}_1, \mathbb{T}_2, \ldots, \mathbb{T}_n\}$; where $\mathbb{T}_i$ is the $i^{th}$ task held by $\mathbb{R}_i$ task requester. The set-up, where each \emph{task requester} carrying multiple tasks is reserved for our future work. Also, along with a task, each task requester $\mathbb{R}_i \in \mathbb{R}$ has an upper bound on the amount he/she (henceforth he) can pay for getting his task executed, known as \emph{budget} represented as $\mathbb{B}_i$. The budget vector for all the task requesters is given as $\mathbb{B}=\{\mathbb{B}_1, \mathbb{B}_2, \ldots, \mathbb{B}_n\}$.  Each of the task requester submits the endowed task along with their publicly known budget to the \emph{platform}. The \emph{platform} projects these tasks to the IoT devices present on the other side of the market. In our set-up, we have \emph{m} IoT devices represented by the set $\mathbb{E} = \{\mathbb{E}_1,\mathbb{E}_2, \ldots, \mathbb{E}_m\}$. It is considered that $m \gg n$. Afterwards, each IoT device shows its interest over the set of tasks for execution purpose to the platform along with the maximum value it can charge for executing each task. Utilizing the submitted information by the IoT devices, we can have the set of IoT devices that are interested to execute the task $\mathbb{T}_j$ and is given as $\mathbb{\boldsymbol{I}}^j = \{\mathbb{E}_1,\mathbb{E}_2, \ldots, \mathbb{E}_{k_j}\}$; where $k_j$ is the number of IoT devices associated with task $\mathbb{T}_j$. The set $\mathbb{\boldsymbol{I}}$ = $\{\mathbb{\boldsymbol{I}}^1, \mathbb{\boldsymbol{I}}^2, \ldots, \mathbb{\boldsymbol{I}}^n\}$ represents the associated set of IoT devices for all the \emph{n} tasks. The maximum value an IoT device $\mathbb{E}_i$ will charge for executing a task $\mathbb{T}_j$ is given as $v_i^j$ called the valuation. The valuations of the IoT devices are \emph{private} in nature. It is to be noted that the IoT devices are \emph{strategic} in nature. By \emph{strategic} we mean that the IoT devices can misreport their private valuation in order to gain. So, it is better to represent the bid value of each IoT device $\mathbb{E}_i$ for executing the task $\mathbb{T}_j$ as $b_{i}^j$. $b_{i}^{j} = v_{i}^{j}$ represents the fact that the IoT device $\mathbb{E}_i$ report its private valuation $b_{i}^{j}$ for the task $\mathbb{T}_j$ in a \emph{truthful} manner. The bid vector for each task $\mathbb{T}_j$ is given as $b_j = \{b_{1}^{j}, b_{2}^{j}, \ldots, b_{k_j}^{j}\}$. The set $b = \{b_1, b_2, \ldots, b_n\}$ represents the set of bid vectors of the IoT devices for all the tasks. Based on the set $\mathbb{\boldsymbol{I}}$, a non-conflict graph $\mathbb{G}(\mathcal{V},\mathcal{E})$ is constructed; where $\mathcal{V}$ is the set of vertices representing the tasks. An edge $(i,~j) \in \mathcal{E}$ between the tasks \emph{i} and \emph{j} represents the fact that the pair $(i,~j)$ have at least one IoT device that is associated to both the tasks. Once the graph is constructed, next target is to place the tasks along with their respective IoT devices to different time slots so as to preserve the assumptions made. The set of time slots to which all the tasks are placed in, is given as $\tau = \{1,2, \ldots, \kappa\}$; where $\kappa$ is the number of time slots available. Afterwards, in \emph{peer grading} phase, each IoT device $\mathbb{E}_i$ provides a ranked list over the subset of IoT devices associated with task $\mathbb{T}_j$ denoted by $\succ_{i}^{j}$, where $\mathbb{E}_\ell$ $\succ_{i}^{j}$ $\mathbb{E}_k$ means that the IoT device $\mathbb{E}_i$ ranks $\mathbb{E}_\ell$ above $\mathbb{E}_k$. For each task $\mathbb{T}_j$, this \emph{peer grading} process will result in the quality IoT devices. Now, the next target is to select the subset of IoT devices from the quality IoT devices for each task and decide their payment. The allocation vector for all the tasks is given as $\mathbb{A}=\{\mathbb{A}_1, \mathbb{A}_2, \ldots, \mathbb{A}_n\}$; where $\mathbb{A}_i$ contains the IoT devices selected for task $\mathbb{T}_i$. Similarly, the payment vector of all the IoT devices for \emph{n} tasks is given as $\mathbb{\boldsymbol{P}} = \{\mathbb{\boldsymbol{P}}_1, \mathbb{\boldsymbol{P}}_2, \ldots, \mathbb{\boldsymbol{P}}_n\}$. Here, $\mathbb{\boldsymbol{P}}_j$ is the payment vector of IoT devices associated with task $\mathbb{T}_j$ and is given as $\mathbb{\boldsymbol{P}}_j = \{\mathbb{\boldsymbol{P}}_1^{j}, \ldots, \mathbb{\boldsymbol{P}}_{k_j}^{j}\}$; where $\mathbb{\boldsymbol{P}}_i^j$ is the payment received by IoT device $\mathbb{E}_i$ for executing task $\mathbb{T}_j$. The utility achieved by any $i^{th}$ IoT device for each task $\mathbb{T}_j$ could be defined as the payment it received for executing task $\mathbb{T}_j$ minus the valuation of an IoT device for task $\mathbb{T}_j$, if it is considered for task $\mathbb{T}_j$; otherwise 0. This can be represented formally as:
    \begin{equation}                   
 u_i^j =
  \begin{cases}
  \mathbb{\boldsymbol{P}}_i^j - v_{i}^{j}, & \textit{if $\mathbb{E}_i$ is considered for task $\mathbb{T}_j$} \\
   0,         & \textit{Otherwise}
  \end{cases}
  \end{equation}
\begin{definition}[Incentive Compatible (IC) \cite{NNis_Pre_2007}] A mechanism is said to be truthful or IC if reporting true valuation by any agent \emph{i} will maximize its utility irrespective of the valuations of other agents. Formally in our case, for any arbitrary IoT device $\mathbb{E}_i$ for task $\mathbb{T}_j$ the utility relation is $u_{i}^{j} = \mathbb{\boldsymbol{P}}_i^j - v_{i}^{j} \geq \mathbb{\boldsymbol{P}}_i^j - b_{i}^{j} = \hat{u}_{i}^{j}$; where $u_{i}^{j}$ is the utility when $\mathbb{E}_i$ reports true value and $\hat{u}_{i}^{j}$ is the utility when reporting the bid other than the true value $b_{i}^j \neq v_i^j$.

\end{definition}

\begin{definition}[Individual Rationality (IR) \cite{NNis_Pre_2007}]
A mechanism is said to be individually rational if every agent \emph{i} results in a non-negative utility. More formally in our case, $u_{i}^{j} \geq 0$ when participating in the system  
\end{definition}

\begin{definition}[Budget Feasibility (BF) \cite{Singer_2016}]
A mechanism is said to be budget feasible if the total payment made to the agents are within total budget. More formally in our case, fix a task $\mathbb{T}_j$ we have, $\sum\limits_{i=1}^{k_j} \mathbb{\boldsymbol{P}}_{i}^j \leq \mathbb{B}_j$. 
\end{definition}

\section{Proposed Mechanism: TUBE-TAP}\label{s:pm}
In this section, we have proposed a \emph{truthful} mechanism namely TUBE-TAP for the problem under investigation. The main components of the TUBE-TAP are: \emph{Time slot allocation heuristic}, \emph{Quality determination rule}, and \emph{Allocation and payment rule}.
\subsection{Time slot allocation heuristic}\label{sub:tsah}
The underlying idea behind proposing \emph{Time Slot Allocation Heuristic} motivated by\footnote{https://www.youtube.com/watch?v=dJfQQNY7NdU} is to distribute the tasks into different time slots, so that: (a) the IoT devices gets the privilege to execute all the tasks for which they have shown their interest; (b) each IoT device executes a single task at a time.  

\subsubsection{Outline of Time slot allocation heuristic}\label{subsec:1}
\begin{mdframed}[backgroundcolor=gray230]
\begin{center}\textbf{\underline{Time slot allocation heuristic}}\end{center}
\paragraph*{\textbf{First Phase:}}
\begin{enumerate}
\item Pick a task $\mathbb{T}_i$ which has less than $\kappa$ adjacent tasks in a graph $\mathbb{G}$.
\item Put $\mathbb{T}_i$ on the stack and remove it along with the incident edges from the graph $\mathbb{G}$.
\item Repeat step 1 and 2, until the graph $\mathbb{G}$ is non-empty.  
\end{enumerate}
\paragraph*{\textbf{Second Phase:}}
In each iteration:
\begin{enumerate}
\item Pop the task present at the top of the stack. 
\item Assign it the lowest numbered time slot that is not assigned to any of its neighbouring tasks.
\end{enumerate}
\end{mdframed}
\subsubsection{Detailed Time slot allocation heuristic}\label{subsec:2}
This section explains the detailing of the \emph{Time Slot Allocation Heuristic} presented in the Algorithm \ref{algo:1}. As in the outline of the \emph{Time Slot Allocation Heuristic} in subsection \ref{subsec:1}, it is discussed that it is a two phase mechanism. The first phase of the mechanism is depicted in line $2-9$ of Algorithm \ref{algo:1}. In each iteration of \emph{while} loop in line $2-9$, a task with neighbours less than the $\kappa$ ($\kappa$ time slots are available) is picked-up and is pushed into the stack $S$. Next, the recently pushed task is removed from the graph $\mathbb{G}$ along with its incident edges. In the second phase, shown in line $10-14$ of Algorithm \ref{algo:1}, the actual process of time slots allocation is carried out. For each iteration of \emph{while} loop in line 10-14, the currently present top element is popped out of the stack $S$ and held in $\Bbbk$ data structure. The element held in $\Bbbk$ data structure is added back to graph $\mathbb{G}$.
\IncMargin{0.2em}
\begin{algorithm}[H]\label{algo:1}
\DontPrintSemicolon
\SetNoFillComment

\caption{Time slot allocation heuristic ($\mathbb{G}$, $\kappa$)}
	$\mathbb{G}' \leftarrow \mathbb{G}$, \emph{S} $\leftarrow$ $\phi$ \\
	\While{$\mathbb{G} \neq \phi$}
	{
	\ForEach{$\mathbb{T}_j$ $\in$ $\mathcal{V}$}
	{
          \If{$|adj(\mathbb{T}_j)|$ $<$ $\kappa$}
          {
        Push($S$, $\mathbb{T}_j$) \tcp*{Task $\mathbb{T}_j$ is pushed into the stack $S$}
        $\mathbb{G} \leftarrow \mathbb{G} \setminus \{\mathbb{T}_j\}$ \tcp*{Task $\mathbb{T}_j$ is removed from $\mathbb{G}$}
        }
        }
        }
        \While{$S \neq \phi$}
        {
                  $\Bbbk$ $\leftarrow$ Pop($S$)  \tcp*{$\Bbbk$ holds an element popped-up from stack $S$}
                  $\mathbb{G} \leftarrow \mathbb{G} \cup \{\Bbbk\}$ \tcp*{Construct graph $\mathbb{G}$ by utilizing the neighbours information from $\mathbb{G}'$}
                  Assign $\Bbbk$ the lowest numbered time slot that is not assigned to any of its neighbours.
        }
                       
        \Return $\mathbb{G}$\\
\end{algorithm}
\IncMargin{0.2em}
Each time a task is added in a graph $\mathbb{G}$ the information about neighbouring tasks is fetched from $\mathbb{G}'$ graph. Now, the task added in current iteration is assigned a lowest numbered time slot that is not assigned to its neighbours using line 13. The \emph{while} loop terminates once the stack is empty, or in other words each task is assigned a time slot. Finally, in line 15 a graph $\mathbb{G}$ containing the information about the assigned time slot to each of the task is returned.   
\begin{example}\label{sub:ncgf}
For the understanding purpose, we have considered 5 tasks and 20 IoT devices. 
\begin{figure}[H]
        \begin{subfigure}[b]{0.40\linewidth}
        \centering
                \includegraphics[scale=0.65]{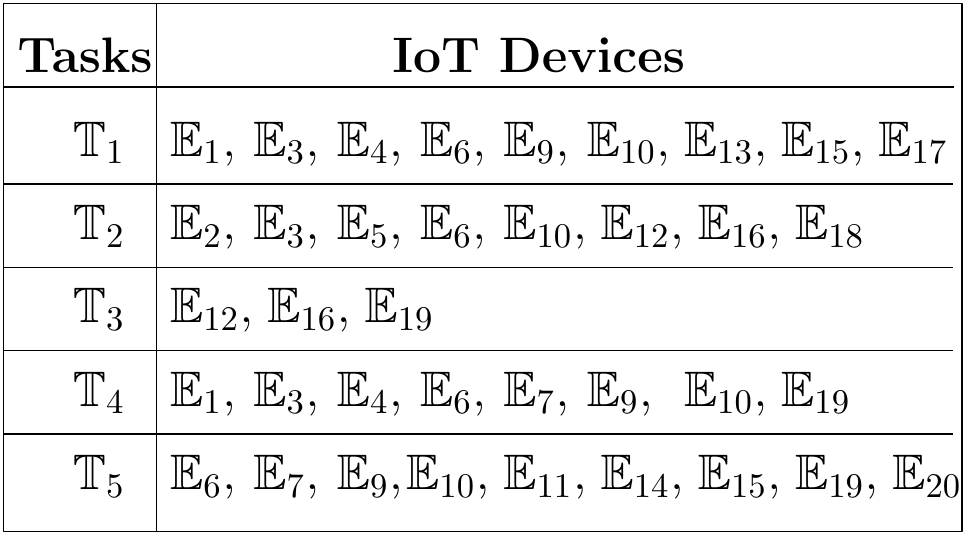}
                \caption{IoT devices showing interests over tasks}
                \label{fig:1(a)}
        \end{subfigure}%
        ~~~
         \begin{subfigure}[b]{0.30\linewidth}
        \centering
                \includegraphics[scale=0.7]{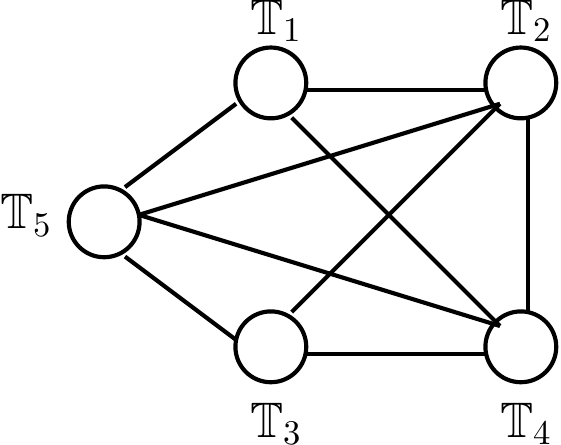}
                \caption{Non-conflict graph}
                \label{fig:1(b)}
        \end{subfigure}%
        ~~~
        \begin{subfigure}[b]{0.30\linewidth}
        \centering
                \includegraphics[scale=0.7]{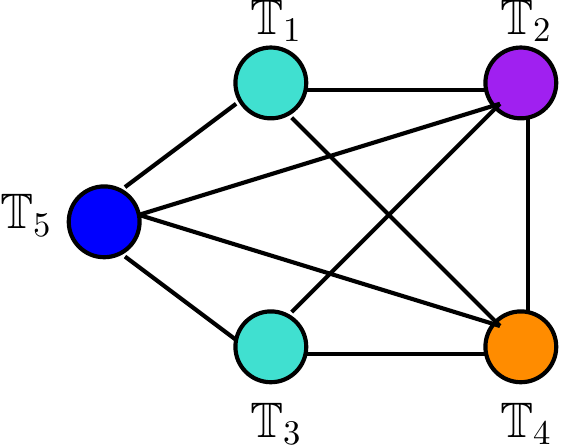}
                \caption{Time slot allocation to tasks}
                \label{fig:1(c)}
        \end{subfigure}
        \caption{Detailed illustration of Algorithm \ref{algo:1}}\label{fig:1}
\end{figure}
\indent Let the budget associated with the 5 tasks are: $\mathbb{B}_1 = 50\$$, $\mathbb{B}_2 = 25\$$, $\mathbb{B}_3 = 30\$$, $\mathbb{B}_4 = 60\$$, and $\mathbb{B}_5 = 15\$$. For each task, the interested set of IoT devices is depicted in Figure \ref{fig:1(a)}. Figure \ref{fig:1(a)} will be read as, say for example consider task $\mathbb{T}_3$. The IoT devices that are interested to execute task $\mathbb{T}_3$ are $\mathbb{E}_{12}$, $\mathbb{E}_{16}$, and $\mathbb{E}_{19}$. Based on the configuration shown in Figure \ref{fig:1(a)}, a graph $\mathbb{G}$ is formed as shown in Figure \ref{fig:1(b)}. Note that the tasks $\mathbb{T}_1$ and $\mathbb{T}_3$ do not share any common IoT devices, so they do not have an edge between them. The result of which they can be placed in the same time slot. In our case the tasks $\mathbb{T}_1$ and $\mathbb{T}_3$ belong to the same time slot, say time slot 1. Tasks $\mathbb{T}_2$, $\mathbb{T}_4$, and $\mathbb{T}_5$ share a common IoT devices so they have an edge between them and will be placed in three different time slots. Also, these tasks have an edge with $\mathbb{T}_1$ and $\mathbb{T}_3$ so they can not be placed in time slot 1. Let the task $\mathbb{T}_2$, $\mathbb{T}_4$, and $\mathbb{T}_5$ are placed in time slot 2, time slot 3, and time slot 4 respectively.
\end{example}
\subsection{Quality Determination Mechanism}\label{subsec:2b}
As the quality of the IoT devices are unknown, in this section a mechanism is proposed for determining the quality of the IoT devices. First, the outline of the \emph{Quality Determination Mechanism} is presented in sub section \ref{subsec:qdr} and in sub section \ref{subsec:qdrd} the detailed version of the mechanism is
discussed. 
\subsubsection{Outline of The Quality Determination Mechanism}\label{subsec:qdr}
\begin{mdframed}[backgroundcolor=gray230]
\begin{center}\textbf{\underline{Quality Determination Mechanism}}\end{center}
\noindent \textbf{Repeat:}
\begin{enumerate}
\item For each task $T_i$, assign \emph{r} IoT devices to $r'$ other IoT devices for the ranking purpose; here $r' \gg r$.
\item Select an IoT device that appears at first place in most of the rankings.
\end{enumerate}
\textbf{Until:} Each IoT device is considered for the ranking.  
\end{mdframed}
\subsubsection{Detailed Quality Determination Mechanism}\label{subsec:qdrd}
This section presents the detailing of the \emph{Quality Determination Mechanism}. Prior to this mechanism, the \emph{Main Routine} is presented in Algorithm \ref{algo:4}. The idea behind providing the \emph{Main Routine} is to capture each task of the system present in different time slots.

\begin{algorithm}[H]\label{algo:4}
\DontPrintSemicolon
\SetNoFillComment
    \SetKwInOut{Output}{Output}
\caption{Main Routine ($\mathbb{G}$, $\mathbb{B}$, $\mathbb{\boldsymbol{I}}$, $\tau$, $\mathbb{T}$, $b$)}
      \Output{ $\mathbb{A}$, $\mathbb{\boldsymbol{P}}$}
	\ForEach{$i \in \tau$}
	{
	    \ForEach{$\mathbb{T}_j \in i$}
	    {
	    ($\pi^j$, $\tilde{b}_j$) $\leftarrow$ Quality Determination Mechanism ($\mathbb{T}_j$, $\mathbb{\boldsymbol{I}}^j$)\\
	   ($\mathbb{A}'_j,\mathbb{\boldsymbol{P}}'_j$) $\leftarrow$  Allocation and Payment Rule ($\pi^j$, $\tilde{b}_j$, $\mathbb{B}_j$)\\
	   $\mathbb{A} \leftarrow \mathbb{A} \cup \mathbb{A}'_j$\\
	   $\mathbb{\boldsymbol{P}} \leftarrow \mathbb{\boldsymbol{P}} \cup \mathbb{\boldsymbol{P}}'_j$\\
	}
	}
        \Return $\mathbb{A}$, $\mathbb{\boldsymbol{P}}$\\	     
\end{algorithm}
In main routine, line $1-8$ keeps track of each time slot and in each time slot each task is taken care by line $2-7$. Line 9 returns the allocation and payment vectors for all the tasks in the system. In Algorithm \ref{algo:2}, initialization of data structures are done in line 1. In line 2, $\Psi'_j$ and $\Psi_j$ data structures keeps the copy of the IoT devices that execute the task $\mathbb{T}_j$. 
The \emph{do while} loop in line 3-14 iterates until all the IoT devices got ranked. Using line 4, \emph{r} random IoT devices are picked up that are to be ranked and stored in the data structure $\Psi$. Similarly, in line 5, the $r'$ IoT devices other than that are selected by line 4 of Algorithm \ref{algo:2} are considered for the ranking process and stored in data structure $\varphi$. Here, $r' \gg r$. Line 6 assigns the completed task of each IoT device in set $\Psi$ to each of the IoT device $\mathbb{E}_i$ in $\varphi$ for ranking purpose.
  
\begin{algorithm}[H]\label{algo:2}
\DontPrintSemicolon
\SetNoFillComment
        \SetKwRepeat{Do}{do}{while}
    \SetKwInOut{Output}{Output}
\caption{Quality Determination Mechanism ($\mathbb{T}_j$, $\mathbb{\boldsymbol{I}}^j$)}
      \Output{ $\mathbf{\Phi}_j \leftarrow \phi$}
	$\Psi \leftarrow \phi$, $\varphi \leftarrow \phi$, $\mathcal{N}' \leftarrow \phi$, $\beta \leftarrow \phi$\\
	         $\Psi'_j = \Psi_j = \mathbb{\boldsymbol{I}}^j$ \tcp*{$\Psi'_j$ and $\Psi_j$ keeps the copy of IoT devices that executes $\mathbb{T}_j$.}
	         \Do{$\Psi_j \neq \phi$}
	         {
	         $\Psi$ $\leftarrow$ Pick\_random ($\Psi_j$, $r$) \tcp*{ Pick \emph{r} IoT devices from $\Psi_j$.}
	         $\varphi$ $\leftarrow$ Pick\_random ($\Psi'_j \setminus \Psi$, $r'$) \tcp*{Pick $r'$ IoT devices from $\Psi'_j \setminus \Psi$.}
	         Assign the completed task $\mathbb{T}_j$ of each IoT devices in $\Psi$ to the IoT devices in $\varphi$.\\
	         \ForAll{$\mathbb{E}_i \in \varphi$}
	         {
	           $\beta \leftarrow$ Select\_best($\succ_i^j$) \tcp*{Select top ranked IoT device from ${\mathbb{E}_i}'s$ ranked list for task $\mathbb{T}_j$ given as $\succ_i^j$.}
	           $\mathcal{N}' \leftarrow \mathcal{N}' \cup \{\beta\}$ \tcp*{$\mathcal{N}'$ data structure allows the duplication of elements.}
	         }
	         $\mathbf{\Phi}_j \leftarrow \mathbf{\Phi}_j \cup \{\max\limits_{\mathbb{E}_k \in \mathcal{N}'}\{|S_k|\}\}$ \tcp*{$S_k$ is the set of ${E_k}'s$ in $\mathcal{N}'$.}
	         $\tilde{b}_j$ $\leftarrow$ $\tilde{b}_j$ $\cup$ $\{b_k^j\}$ \tcp*{$\tilde{b}_j$ maintains the bid values of the quality IoT devices.}
	         $\Psi_j \leftarrow \Psi_j \setminus \Psi$          
	           } 
        \Return $\mathbf{\Phi}_j$, $\tilde{b}_j$ 
\end{algorithm}
Using line 7-10 for each iteration of \emph{for} loop record about the top ranked IoT device by each $\mathbb{E}_i \in \varphi$ is kept in the $\mathcal{N}'$ data structure. In line 11, $\mathbf{\Phi}_j$ data structure captures the IoT device that was ranked top by most of the IoT devices for task $\mathbb{T}_j$. Line 13 removes the IoT devices that are ranked in the current iteration from $\Psi_j$. Finally, line 14 returns $\mathbf{\Phi}_j$ that contains the quality IoT devices for task $\mathbb{T}_j$.
 
 \begin{example}\label{sub:ncgf2}
For the detailed illustration of Algorithm \ref{algo:2} we have considered the set-up discussed
\begin{figure}[H]
        \begin{subfigure}[b]{0.33\linewidth}
        \centering
                \includegraphics[scale=0.5]{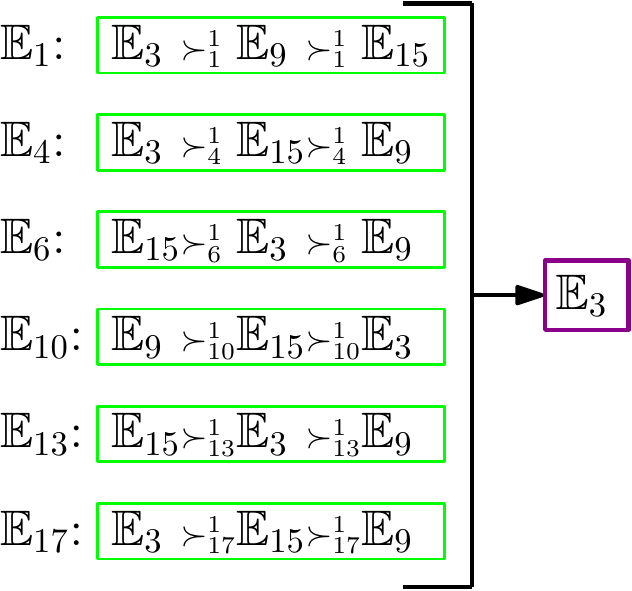}
                \caption{Peer grading ($1^{st}$ iteration)}
                \label{fig:3(a)}
        \end{subfigure}%
        ~~~
         \begin{subfigure}[b]{0.33\linewidth}
        \centering
                \includegraphics[scale=0.5]{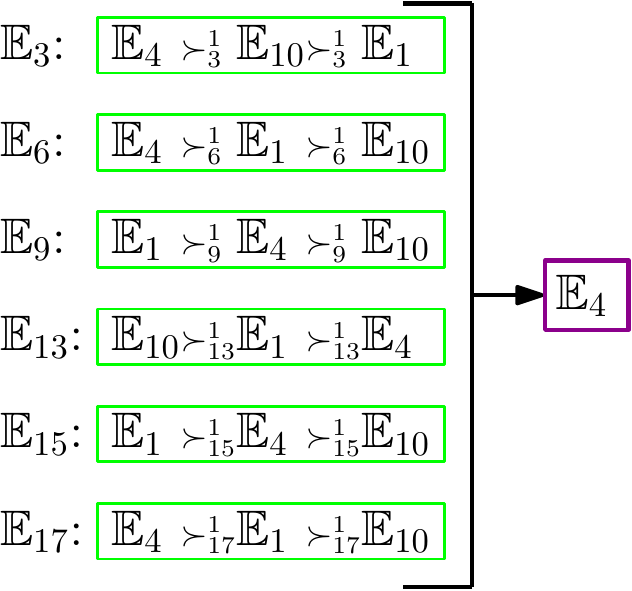}
                \caption{Peer grading ($2^{nd}$ iteration)}
                \label{fig:3(b)}
        \end{subfigure}%
        ~~~
        \begin{subfigure}[b]{0.33\linewidth}
        \centering
                \includegraphics[scale=0.5]{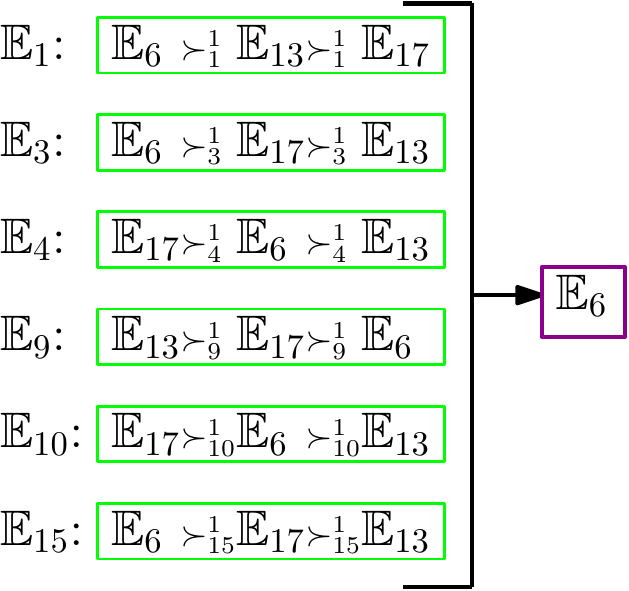}
                \caption{Peer grading ($3^{rd}$ iteration)}
                \label{fig:3(c)}
        \end{subfigure}
        \caption{Detailed illustration of Algorithm \ref{algo:2}}\label{fig:1b}
\end{figure}
 in Example \ref{sub:ncgf}.
 In this example, we have illustrated Algorithm \ref{algo:2} for one task, say task $\mathbb{T}_1$. However, one can follow the similar procedure for the remaining tasks. For the $1^{st}$ iteration of the peer grading process, we have randomly selected 3 IoT devices ($r=3$) say $\mathbb{E}_3$, $\mathbb{E}_9$, and $\mathbb{E}_{15}$ and assigned to the remaining IoT devices for the grading purposes. Next, following the Algorithm \ref{algo:2}, we have to check which IoT device among $\mathbb{E}_3$, $\mathbb{E}_9$, and $\mathbb{E}_{15}$ has been top ranked by the majority of the peers. From Figure \ref{fig:3(a)} one can see that $\mathbb{E}_3$ has been top ranked by the majority of the peers. So, for the time being $\mathbf{\Phi}_1 = \{\mathbb{E}_3\}$. In the similar fashion, we can follow the other iterations of the peer grading process as shown in Figure \ref{fig:3(b)} and Figure \ref{fig:3(c)} and determine the quality IoT devices. 
 At the end of the peer grading process, the set of quality IoT devices for task $\mathbb{T}_1$ is given as $\mathbf{\Phi}_1 = \{\mathbb{E}_3, \mathbb{E}_4, \mathbb{E}_6\}$.
\end{example}
\subsection{Allocation and Payment Rule}\label{subsec:3}

\begin{algorithm}[H]\label{algo:5}
\DontPrintSemicolon
\SetNoFillComment
    \SetKwInOut{Output}{Output}
\caption{Allocation and Payment Rule ($\pi^j$, $\tilde{b}_j$, $\mathbb{B}_j$)}
      \Output{ $\mathbb{A}_j$, $\mathbb{\boldsymbol{P}}_j$}
	\tcc{Allocation Rule}
	Sort($\pi^j$, $\tilde{b}_j$) \tcp*{Sort $\pi^j$ based on $\tilde{b}_j$ as $b_{1}^{j} \leq b_{2}^{j} \leq \ldots \leq b_{\tilde{k}_j}^{j}$; such that $\tilde{k}_j < k_j$}
	$k \leftarrow 1$\\
	\While{$b_{i}^{j} \leq \frac{\mathbb{B}_j}{k}$}
	{
	     $\mathbb{A}_j \leftarrow \mathbb{A}_j \cup \{\mathbb{E}_i\}$\\
	     $k \leftarrow k + 1$  
        }
        \tcc{Payment Rule}
        \ForEach{$\mathbb{E}_i \in \mathbb{A}_j$}
        {
           $\mathbb{\boldsymbol{P}}_i^{j} \leftarrow  \{min\{\frac{\mathbb{B}_j}{k},~b_{k+1}^{j}\}\}$\\
           $\mathbb{\boldsymbol{P}}_j \leftarrow \mathbb{\boldsymbol{P}}_j \cup \{\mathbb{\boldsymbol{P}}_i^{j}\}$  
        }
        \Return $\mathbb{A}_j$, $\mathbb{\boldsymbol{P}}_j$\\	     
\end{algorithm}
  
This section explains the \emph{Allocation and Payment Rule} presented in the Algorithm \ref{algo:5}.
Considering the allocation rule, in line 1 first the quality IoT devices in $\pi^j$ is sorted in increasing order based on the bid vector $\tilde{b}_j$. The variable \emph{k} is initialized to 1.   
The \emph{while} loop in line $3-6$ determines the largest index \emph{k} that satisfies the stopping condition of the \emph{while} loop. The $\mathbb{A}_j$ data structure in line 4 keeps track of winning IoT devices. Talking about the payment rule, for each $\mathbb{E}_i$ in $\mathbb{A}_j$ the minimum among $\frac{\mathbb{B}_j}{k}$ and $b_{k+1}^{j}$ is taken as the payment. Finally, line 11 returns the allocation and payment for the task $\mathbb{T}_j$.

\begin{example}\label{sec:ex}
\label{s:ie}
For understanding the allocation and payment rule, let us continue with the quality IoT devices resulted from Example \ref{sub:ncgf2}. The budget given for task $\mathbb{T}_1$ is 50 \$. The quality IoT devices along with their bid values is depicted in Figure \ref{fig:3a}. Utilizing Algorithm \ref{algo:5} in the set-up shown in  Figure \ref{Fig:3a}, first the IoT devices are sorted in decreasing order of their bid value as shown in Figure \ref{fig:3b}. In our case, from the ordering, first $\mathbb{E}_4$ is picked up and considered as the check $10 \leq \frac{50}{1}$ is satisfied for $\mathbb{E}_4$. Next, $\mathbb{E}_3$ is picked up from the ordering and is also considered because of the similar reason. Next, $\mathbb{E}_6$ is picked up from the ordering and will be not be considered as the check $30 \leq \frac{50}{3}$ is not satisfied. 
 \begin{figure}[H]
\begin{subfigure}[b]{0.33\textwidth}
\begin{center}
                \includegraphics[scale=0.6]{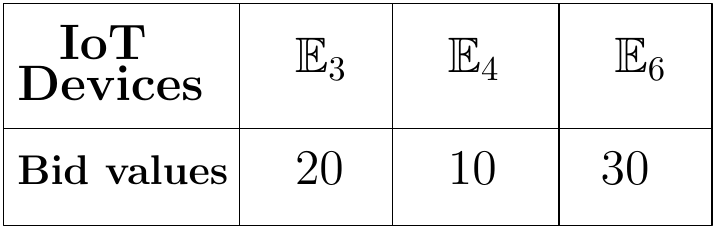}
                \caption{Bid value configuration}
                \label{fig:3a}
                \end{center}
        \end{subfigure}%
        \begin{subfigure}[b]{0.33\textwidth}
        \begin{center}
                \includegraphics[scale=0.6]{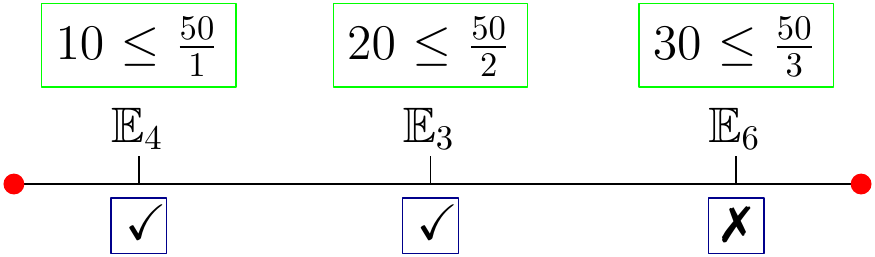}
                \caption{Allocation resulted}
                \label{fig:3b}
                \end{center}
        \end{subfigure}
        \begin{subfigure}[b]{0.33\textwidth}
        \begin{center}
                \includegraphics[scale=0.6]{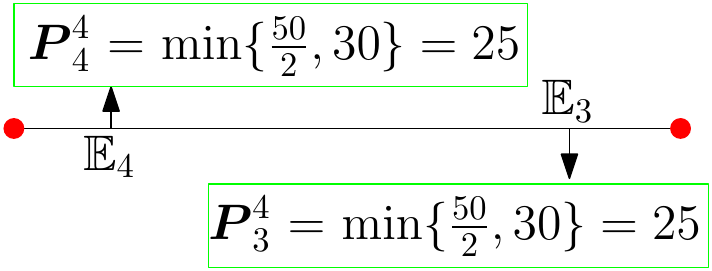}
                \caption{Payment determination}
                \label{fig:3c}
                \end{center}
        \end{subfigure}
        \caption{Detailed illustration of Algorithm \ref{algo:5} (case with $\frac{\mathbb{B}_j}{k}$ as payment)}
        \label{Fig:3}
\end{figure}    
 So, we have $\mathcal{A}_1 = \{\mathbb{E}_4, \mathbb{E}_3\}$ as the winning set. So, we get the \emph{k} value as 2 for our example. Next, the payment calculation of the $\mathbb{E}_4$ and $\mathbb{E}_3$ is presented in Figure \ref{fig:3c}. For $\mathbb{E}_4$ we have $\mathbb{\boldsymbol{P}}_4^1 = min\{\frac{50}{2},~30\}= 25$, and for $\mathbb{E}_3$ we have $\mathbb{\boldsymbol{P}}_3^1 = min\{\frac{50}{2},~30\}= 25$.
\end{example}
\begin{example}
 As in the above example, it can be seen that the payment for both the IoT devices is the left term of the payment rule, so the remaining budget is zero. In order to see when the right term of the payment rule will be coming into picture the example in Figure \ref{Fig:3} is repeated for different bid configuration in Figure \ref{Fig:4}.   
 \begin{figure}[H]
\begin{subfigure}[b]{0.33\textwidth}
\begin{center}
                \includegraphics[scale=0.6]{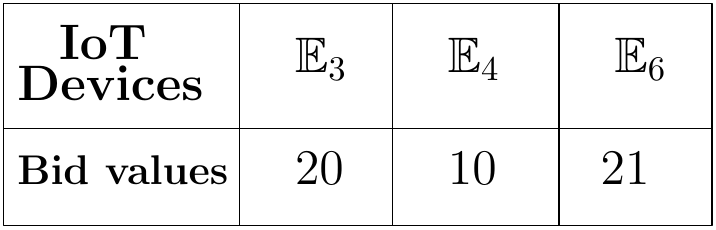}
                \caption{Bid value configuration}
                \label{fig:4a}
                \end{center}
        \end{subfigure}%
        \begin{subfigure}[b]{0.33\textwidth}
        \begin{center}
                \includegraphics[scale=0.6]{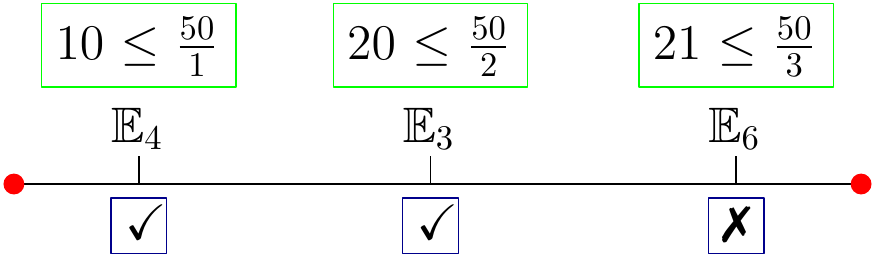}
                \caption{Allocation resulted}
                \label{fig:4b}
                \end{center}
        \end{subfigure}
        \begin{subfigure}[b]{0.33\textwidth}
        \begin{center}
                \includegraphics[scale=0.6]{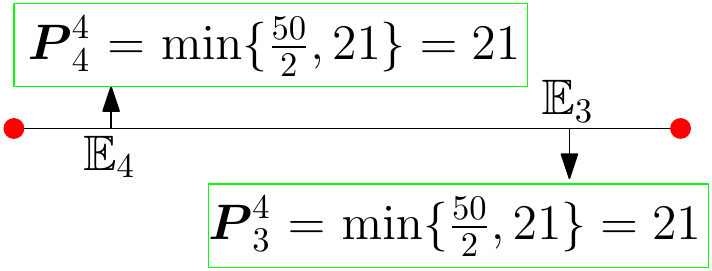}
                \caption{Payment determination}
                \label{fig:4c}
                \end{center}
        \end{subfigure}
        \caption{Detailed illustration of Algorithm \ref{algo:5} (case with $b_{k+1}^j$ as payment)}
        \label{Fig:4}
\end{figure}  
In this example, the allocation set will be similar to what we obtained for Example \ref{sec:ex} as shown in Figure \ref{fig:4b}.   
The payment calculation of the IoT devices $\mathbb{E}_4$ and $\mathbb{E}_3$ is presented in Figure \ref{fig:4c}. For $\mathbb{E}_4$ we have $\mathbb{\boldsymbol{P}}_4^1 = min\{\frac{50}{2},~21\}= 21$, and for $\mathbb{E}_3$ we have $\mathbb{\boldsymbol{P}}_3^1 = min\{\frac{50}{2},~21\}= 21$.
\end{example}



\section{Analysis of TUBE-TAP}\label{sec:faba}
This section presents the analysis of TUBE-TAP.

\begin{proposition}
The proposed mechanism in \cite{Singer_2016} has an approximation ratio of 2.
\end{proposition}

\begin{lemma}\label{lemma:1}
TUBE-TAP is truthful.
\end{lemma}
\begin{proof}
The proof is divided into two cases. In the first case, we have taken an arbitrary winning IoT device into consideration and discuss the impact on its gain (or utility), when it \emph{deviates} from its true valuation. In second case, we have considered any arbitrary losing IoT device and analysis similar to Case 1 is done. Fix a task $\mathbb{T}_j$.
\paragraph*{Case 1:} Let us suppose that $i^{th}$ winning IoT device deviates from its true value and reports a bid value $b_{i}^{j} < v_{i}^{j}$. As the IoT device $\mathbb{E}_i$ was winning with $v_{i}^{j}$ it will continue to win with $b_{i}^{j}$ because by reporting value lesser than the true value, it will be appearing early in the ordering. So, its utility will be $\hat{u}_i^j = \mathbb{\boldsymbol{P}}_i^j - v_i^j$ which is same as $u_i^j$. But, if it reports $b_{i}^{j} > v_{i}^{j}$, this gives rise to two possibilities. One possibility could be, it would continue to win by appearing later in the ordering and in that case his utility will be $\hat{u}_i^j = \mathbb{\boldsymbol{P}}_i^j - v_i^j = u_i^j$. Another possibility could be, it may lose by appearing later in the ordering in that case its utility will be $\hat{u}_i^j$ = 0.    
\paragraph*{Case 2:} Let us suppose that $i^{th}$ losing IoT device deviates from its true value and reports a bid value $b_{i}^{j} > v_{i}^{j}$. As the IoT device $\mathbb{E}_i$ was losing with $v_{i}^{j}$ it will continue to lose by $b_{i}^{j}$ because by deviating this way it will be appearing later in the ordering. So, its gain will be $\hat{u}_i^j$ = 0 which is same as $u_i^j$. But, if it reports $b_{i}^{j} < v_{i}^{j}$, then the two possibilities arises. One possibility could be, by deviating this way it could appear early in the ordering but still continue to lose and in that case $\hat{u}_i^j$ = 0 which is same as $u_i^j$. Another possibility could be, it could win, in that case it had defeated the IoT device $\mathbb{E}_k$ with valuation $v_{k}^{j} < v_{i}^{j}$ and hence $b_{i}^{j}<v_{k}^{j}$. In this case, its payment will be less as compared to its true valuation. So, its utility $\hat{u}_i^j = \mathbb{\boldsymbol{P}}_i^j -  v_i^j < 0$. Hence, no gain is achieved.\\
\\
\indent Considering Case 1 and Case 2, it can be concluded that the IoT devices cannot gain by misreporting their true value. So, TUBE-TAP is truthful.  
\end{proof}

\begin{lemma}\label{lemma:2}
In TUBE-TAP, for each task requester $\mathbb{R}_j$ the total payment $\mathbb{\boldsymbol{P}}_j$ made to the IoT devices are within available budget $\mathbb{B}_j$. More formally, $\mathbb{\boldsymbol{P}}_j$ = $\sum\limits_{\mathbb{E}_i \in \mathcal{A}_j} \mathbb{\boldsymbol{P}}_i^j \leq \mathbb{B}_j$. Also, $\sum\limits_{\mathcal{A}_j \in \mathcal{A}} \sum\limits_{\mathbb{E}_i \in \mathcal{A}_j} \mathbb{\boldsymbol{P}}_i^j \leq \sum\limits_{\mathbb{T}_j \in \mathbb{T}}\mathbb{B}_j$.    
\end{lemma}
\begin{proof}
Fix a task requester $\mathbb{R}_j$ and a task $\mathbb{T}_j$. From the construction of TUBE-TAP, it is clear that, the maximum payment that any winning IoT device will be paid is $\frac{\mathbb{B}_j}{k}$; where \emph{k} is the largest index obtained in the ordering of IoT devices that satisfies $b_{k}^{j} \leq \frac{\mathbb{B}_j}{k}$. Now, the total payment $\mathbb{\boldsymbol{P}}_j$ is given as:
\begin{equation*}
\mathbb{\boldsymbol{P}}_j = \sum\limits_{\mathbb{E}_i \in \mathcal{A}_j} \mathbb{\boldsymbol{P}}_i^j \leq \sum_{\mathbb{E}_i \in \mathcal{A}_j} \frac{\mathbb{B}_j}{k} = \frac{\mathbb{B}_j}{k} \times k = \mathbb{B}_j 
\end{equation*}
From here we can say that, $\mathbb{\boldsymbol{P}}_j \leq \mathbb{B}_j$. As this is true for any task $\mathbb{T}_j$, so the budget feasibility will hold for all the available tasks $i.e.$ $\sum\limits_{\mathcal{A}_j \in \mathcal{A}} \sum\limits_{\mathbb{E}_i \in \mathcal{A}_j} \mathbb{\boldsymbol{P}}_i^j \leq \sum\limits_{\mathbb{T}_j \in \mathbb{T}}\mathbb{B}_j$. This completes the proof. 
\end{proof}

\begin{lemma}\label{lemma:3}
The allocation resulted by TUBE-TAP is at most 2 allocation away from the optimal one $i.e.$ $OPT \leq 2 \times OM$; where OPT is the optimal allocation and OM is the allocation resulted by TUBE-TAP. 
\end{lemma}
\begin{proof}
Fix a task requester $\mathbb{R}_i$ and task $\mathbb{T}_i$. Let us suppose for the sake of contradiction that the OPT consists of \emph{k} IoT devices $i.e.$ $|OPT|=k$ and OM consists of less than $\frac{k}{2}$ IoT devices $i.e.$ $|OM|< \frac{k}{2}$. It implies that, $b_{\frac{k}{2}}^i > \frac{\mathbb{B}_i}{k/2}$. Note however, that this is impossible since we assume that $b_{\frac{k}{2}}^i \leq \ldots \leq b_{k}^i$, and $\sum_{j=\frac{k}{2}}^{k} b_j^i \leq \mathbb{B}_i$ which implies that $b_{\frac{k}{2}}^i \leq \frac{\mathbb{B}_i}{k/2}$. Hence a contradiction.  
\end{proof}

\begin{lemma}\label{lemma:4}
Let $\mathbb{U}$ be the event given as $\mathbb{U} = \{\mathbb{E}_i$ is considered for task $\mathbb{T}_j$\} and $X_j^i$ is an indicator random variable defined as $X_j^i$ = I$\{\mathbb{U}\}$. Then, the expectation is just the probability of the corresponding event $i.e.$ $E[X_j^i]$ = Pr\{$\mathbb{U}$\} \cite{Coreman_2009}.
\end{lemma}
\begin{proof}
By the definition of indicator random variable, we can write $X_j^i$ is 1 when $\mathbb{U}$ occurs and 0 when $\mathbb{U}$ does not occurs. So, as $X_j^i$ = I\{$\mathbb{U}$\}. Taking expectation both side, we get
\begin{equation*}
E[X_j^i] = E[I\{\mathbb{U}\}]
\end{equation*}
\begin{equation*}
\hspace*{35mm} = 1 \cdot Pr\{\mathbb{U}\} + 0 \cdot Pr\{\mathbb{\bar{U}}\} 
\end{equation*}
\begin{equation*}
E[X_j^i] = Pr\{\mathbb{U}\}
\end{equation*}
where, $\mathbb{\bar{U}}$ denotes $S - \mathbb{U}$ such that $S$ is the sample space. 
\end{proof}

\begin{lemma}\label{lemma:5}
The expected number of times any arbitrary $\mathbb{E}_i$ is considered (or winning) is given as $p \cdot k_i$; where $k_i$ is the number of tasks for which the $i^{th}$ IoT device has shown interest and p is the probability with which $\mathbb{E}_i$ is considered for a task. In other words, $E[X^i] = p \cdot k_i$; where $X^i$ is the random variable measuring the number of times $\mathbb{E}_i$ is considered out of $k_i$. 
\end{lemma}
\begin{proof}
Fix an IoT device $\mathbb{E}_i$, we now wish to compute the expected number of times the $\mathbb{E}_i$ is considered. We capture the total number of times $\mathbb{E}_i$ is considered out of $k_i$ by $X^i$ random variable. So, the expected number of times $\mathbb{E}_i$ is considered is given as $E[X^i]$. Our sample space for $\mathbb{E}_i$ IoT device for any task $\mathbb{T}_j$ is $S$= $\{\mathbb{E}_i$ is considered for task $\mathbb{T}_j$, $\mathbb{E}_i$ not considered for task $\mathbb{T}_j\}$. So, we have Pr\{$\mathbb{E}_i$ is considered for task $\mathbb{T}_j$\}= $p$ and Pr\{$\mathbb{E}_i$ is not considered for task $\mathbb{T}_j$\} = $1-p$.\\
\indent We define the indicator random variable $X_j^i$ as $X_j^i$ = I\{$\mathbb{E}_i$ is considered for task $\mathbb{T}_j$\}; where

 \begin{equation}                   
 X_j^i =
  \begin{cases}
  1, & \textit{if $\mathbb{E}_i$ is considered for task $\mathbb{T}_j$} \\
   0,         & \textit{Otherwise}
  \end{cases}
  \end{equation}
  The expected number of times $\mathbb{E}_i$ is considered for task $\mathbb{T}_j$ is simply the expected value of our indicator random variable $X_j^i$:
  \begin{equation*}
E[X_j^i] = E[I\{\mathbb{E}_i~is~considered~for~task~\mathbb{T}_j\}]
\end{equation*}   
As always with the indicator random variable, the expectation is just the probability of the corresponding event (using lemma \ref{lemma:4}):   
\begin{equation*}
E[X_j^i]= 1 \cdot Pr\{X_j^i=1\} + 0 \cdot Pr\{X_j^i=0\}
\end{equation*}
\begin{equation*}
\hspace*{-15mm}= 1 \cdot p + 0 \cdot (1-p)
\end{equation*}
\begin{equation*}
\hspace*{-35mm}= 1 \cdot p
\end{equation*}
\begin{equation*}
 \hspace*{-50mm} E[X_j^i]= p
\end{equation*}
Now, let us consider the random variable that we are interested in and is given by $X^i = \sum\limits_{j=1}^{k_i} X_j^i$. We can compute $E[X^i]$ by taking expectation both side, we get:	
\begin{equation*}
E[X^i] = E\bigg[\sum_{j=1}^{k_i} X_j^i\bigg] 
\end{equation*} 
By linearity of expectation, we get 
\begin{equation*}
E[X^i] = \sum_{j=1}^{k_i} E[X_j^i] 
\end{equation*} 
From lemma \ref{lemma:4} it can be seen that, the expected value of any random variable is equal to the probability of the corresponding event. So, 
\begin{equation*}
E[X^i] = \sum_{j=1}^{k_i} Pr\{\mathbb{E}_i~is~considered~for~task~\mathbb{T}_j\} 
\end{equation*}
\begin{equation*}
\hspace*{-45mm} = \sum_{j=1}^{k_i} p 
\end{equation*} 
\begin{equation*}
\hspace*{-45mm} = p \cdot k_i.
\end{equation*}
Hence, the claim survived. It is to be noted that if $p = \frac{1}{2}$, then the value of $E[X^i]$ boils down to $\frac{k_i}{2}$. It means that, any arbitrary $\mathbb{E}_i$ in expectation will be considered for half of number of tasks on which it has shown interest.      
\end{proof}

\begin{lemma}\label{lemma:7}
For any arbitrary IoT device $\mathbb{E}_i$ the expected number of longest contiguous rejection out of $k_i$ tasks after which the IoT device is considered is given as $\Theta(\log_p k_i)$. More formally, we can say $E[Y] = \Theta(\log_p k_i)$; where $Y$ is a random variable that captures the longest continuous rejection of any IoT device.   
\end{lemma}
\begin{proof}
Fix an IoT device $\mathbb{E}_i$. In similar line the proof is illustrated in \cite{Coreman_2009}. Our proof is divided into two cases. From Lemma \ref{lemma:5} it can be seen that the probability that $\mathbb{E}_i$ will be considered for any task $\mathbb{T}_j$ is $p$. Let $X^{i}_{kl} = I\{A_{kl}^{i}\}$ be the indicator random variable associated with an event that the IoT device $\mathbb{E}_i$ is rejected for at least \emph{l} tasks starting form $k^{th}$ task. It is to be noted that, the participation in one time slot by the IoT device is independent of the participation in other time slots. So, for any given event $X^{i}_{kl}$, the probability that for all \emph{l} tasks the IoT device is rejected is given as      
\begin{equation}\label{equation:17}
Pr\{A_{kl}^{i}\} = p\cdot p~\cdot\cdot\cdot  l~times = p^l
\end{equation}
As in our case, \emph{k} varies from 1 to $k_i - l +1$ (i.e. $1 \leq k \leq k_i - l +1$), so the total number of such rejections could be formulated as:
\begin{equation*}
Y = \sum\limits_{k=1}^{k_i-l+1} X^{i}_{kl}
\end{equation*}  
Taking expectation both side, we get
 \begin{equation*}
E[Y] = E\bigg[\sum\limits_{k=1}^{k_i-l+1} X^{i}_{kl}\bigg] 
\end{equation*}
By linearity of expectation, we get
 \begin{equation*}
\hspace*{7mm}= \sum\limits_{k=1}^{k_i-l+1} E[X^{i}_{kl}]
\end{equation*}
From the definition of expectation in Lemma \ref{lemma:4}, we have 
\begin{equation*}
E[Y]= \sum\limits_{k=1}^{k_i-l+1} Pr\{A_{kl}^{i}\}
\end{equation*}  
Using equation \ref{equation:17}, we get
\begin{equation*}
 = \sum\limits_{k=1}^{k_i-l+1} p^l 
\end{equation*}
\begin{equation*}
 E[Y] = (k_i-l+1) \cdot p^l
\end{equation*}
Now, for $l = c \log_p k_i$ and for some positive constant $c$, we obtain  
\begin{equation*}
 E[Y] = (k_i-c \log_p k_i+1)\cdot p^{c \log_p k_i} 
\end{equation*}
\begin{equation*}
\hspace*{3mm} = (k_i-c \log_p k_i+1) \cdot k_i^c
\end{equation*}
\begin{equation*}
\hspace*{6mm} = k_i^{c+1} - ck_i^c \log_p k_i + k_i^c
\end{equation*}
\begin{equation*}
\hspace*{-22mm}= \Theta(k_i^c)
\end{equation*}
From here we can conclude that, for some constant $c \geq 1$ the longest continuous rejection boils down to $\Theta(\log_p k_i)$. Hence, the claim survived.
\end{proof}

\begin{lemma}\label{lemma:6}
In our system, the probability that any arbitrary IoT device $\mathbb{E}_i$ is considered (or wins) for at least one time out of $k_i$ is greater than or equal to $1-\frac{1}{e^{p\cdot k_i}}$; where $k_i$ is the number of tasks for which the $i^{th}$ IoT device has shown interest. In other words, $Pr[X^i \geq 1] \geq \bigg(1-\frac{1}{e^{p\cdot k_i}}\bigg)$; where $X^i$ is the random variable measuring the number of times $\mathbb{E}_i$ IoT device is considered out of $k_i$. 
\end{lemma}
\begin{proof}
Fix an IoT device $\mathbb{E}_i$. As $\mathbb{E}_i$ has shown interest on $k_i$ tasks that are present in different time slots. The probability that $\mathbb{E}_i$ will be considered for task $\mathbb{T}_j$ is $p$ (Pr\{$\mathbb{E}_i$ is not considered for task $\mathbb{T}_j$\} = $1-p$). Also, it can be seen that, the consideration of $\mathbb{E}_i$ in any time slot is independent of other time slots. So, the probability that $\mathbb{E}_i$ will not be considered at all for any of the $k_i$ tasks is given as:
\begin{equation*}
Pr[X^i < 1] = (1-p) \cdot (1-p) \ldots k_i~times
\end{equation*} 
\begin{equation*}
\hspace*{-12mm}= (1-p)^{k_i}
\end{equation*}
Following the inequality $1+x \leq e^x$, we get
 \begin{equation*}
Pr[X^i < 1] \leq e^{-p\cdot k_i} = \frac{1}{e^{p\cdot k_i}}
\end{equation*}
Now, the probability that any $\mathbb{E}_i$ will be considered at least once is given as
 \begin{equation*}
Pr[X^i \geq 1] \geq \bigg(1-\frac{1}{e^{p\cdot k_i}}\bigg)
\end{equation*}
Hence, the claim survives. Also, for $p=\ln 2$, we can see that 
\begin{equation*}
Pr[X^i \geq 1] \geq \bigg(1-\frac{1}{e^{\ln 2\cdot k_i}}\bigg)
\end{equation*}
\begin{equation*}
\hspace*{15mm}= \bigg(1-\frac{1}{2k_i}\bigg)
\end{equation*}
It can be concluded that, the term $\frac{1}{2k_i}$ represents that any arbitrary $\mathbb{E}_i$ will not be considered at all is very small, and can say that it is very unlikely to occur. So, the term $(1-\frac{1}{2k_i})$ will be quite large and hence can say that any IoT device could be considered for at least once with larger probability.
\end{proof}



\section{Experimental Findings}\label{sec:ef}
In this section, we measure the efficacy of our proposed mechanism called TUBE-TAP via simulation. It is to be noted that, the TUBE-TAP is compared with the carefully crafted benchmark mechanism that is \emph{non-truthful} in nature. The manipulative behaviour of the IoT devices in case of benchmark mechanism can be seen evidently in the simulation results. It is to be noted that, our benchmark mechanism differs in terms of \emph{allocation} and \emph{payment} policy from the TUBE-TAP.\\
\indent In the benchmark mechanism, for each task, first the IoT devices are sorted in increasing order of their bid value. Afterwards, the IoT devices are picked up sequentially one at a time from the ordering and check is made that: whether \emph{the sum of the valuation of the IoT device next to it in the ordering and some small constant value (say $\epsilon$) is less than or equal to the remaining budget associated with the task} or not. If the stopping condition is satisfied, then the IoT device will be declared as winner, otherwise not. After the declaration of winner set, the payment of any IoT device in the winning set is the sum of the bid value of the IoT device following it in the sorted ordering and the $\epsilon$ value. More formally, the payment of any $i^{th}$ IoT device for the task $\mathbb{T}_j$ is given as $\mathbb{\boldsymbol{P}}_i^j = b_{i+1}^j + \epsilon$; where $b_{i+1}^j$ is the bid value of the IoT device following \emph{i} in the sorted ordering. It is to be noted that the $\epsilon$ value is same throughout the system, it is taken as $\epsilon = 10$ in our case. The unit of bid value and the budget is taken as \$. The experiments are carried out using Python.     
\subsection{Simulation Set-up}
For our simulation purpose, we have varied the number of task requesters and the number of IoT devices so as to analyse the results in a more better sense. Table \ref{tab:1} shows the configuration of different values of number of task requesters and number of IoT devices that has been utilized for the simulation purpose. For each configuration, the experiment runs for 50 rounds ad the required values are plotted by taking average over these 50 rounds. Other than this, in order to strengthen our claim, we have simulated the mechanisms for two different probability distributions independently; namely, \emph{uniform distribution} (UD) and \emph{normal distribution} (ND). Throughout the experiment, the bid value range (in case of UD) for IoT devices and the budget range for the tasks are kept fixed. It is to be noted that, budget is uniformly distributed within the given range for both ND and UD. Considering the case of ND, for generating the bid values of the IoT devices the mean is taken as 110 and standard deviation is taken as 15. 
\begin{table}[H]
\centering
 \caption{Data set utilized for simulation purpose}
\label{data}
\renewcommand{\arraystretch}{1.3}
\scalebox{0.9}{
\begin{tabular}{ l | l  l l l l l}
\hline
\hline
\emph{Task requesters}&50&100&150&200&250&300\\
\hline
\emph{Task executers}&500&1000&1500&2000&2500&3000\\
\hline
\emph{Bid value range} (for UD)&[80,~150]&[80,~150]&[80,~150]&[80,~150]&[80,~150]&[80,~150]\\
\hline
\emph{Budget distribution}&[400,~600]&[400,~600]&[400,~600]&[400,~600]&[400,~600]&[400,~600]\\
\hline
\end{tabular}
\label{tab:1}}
\end{table}
\noindent In order to measure the efficacy of TUBE-TAP, we have taken two performance metrics: 1) Budget utilization, and 2) Utility of the IoT devices.    
\subsection{Result Analysis}
In this section, we are simulating TUBE-TAP which we are claiming is \emph{budget feasible} and \emph{truthful} in our setting against the benchmark mechanism (which will be referred as BM in the figures of simulation results).\\
\indent Considering the first parameter $i.e.$ \emph{Budget utilization}, we can see in Figure \ref{fig:sim1a}, and Figure \ref{fig:sim2a} that the budget utilization in case of TUBE-TAP is a bit more as compared to the budget utilization in case of BM for both ND and UD case. This is due to the fact that, in case of TUBE-TAP each winner is paid a value between the bid value of last winner and the bid value of the first loser present in the sorted ordering. 
\begin{figure}[H]
\begin{subfigure}[b]{0.50\textwidth}
                \centering
                \includegraphics[scale=0.40]{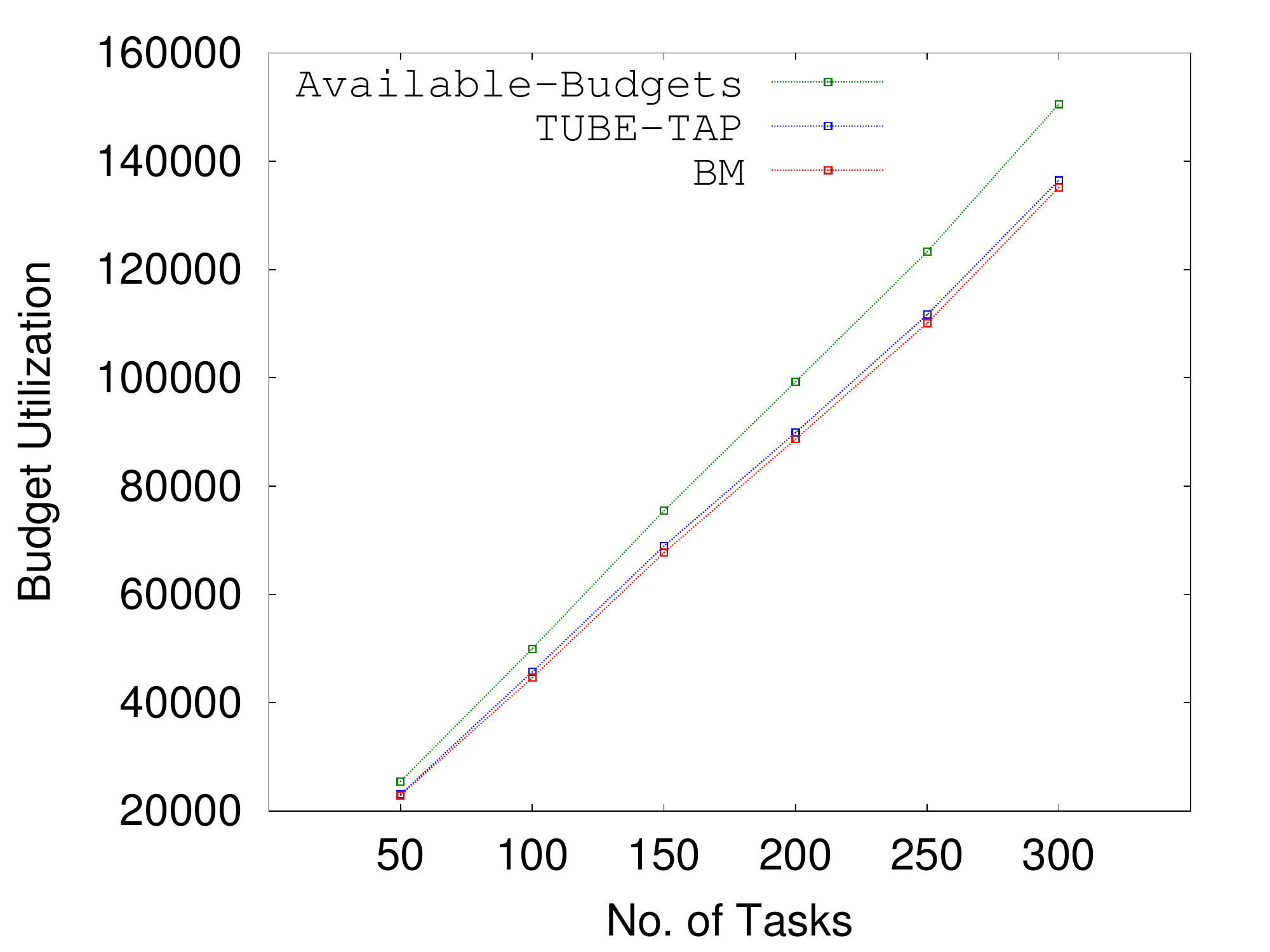}
                \subcaption{Budget utilization (ND)}
                \label{fig:sim1a}
        \end{subfigure}%
        \begin{subfigure}[b]{0.50\textwidth}
                \centering
                \includegraphics[scale=0.40]{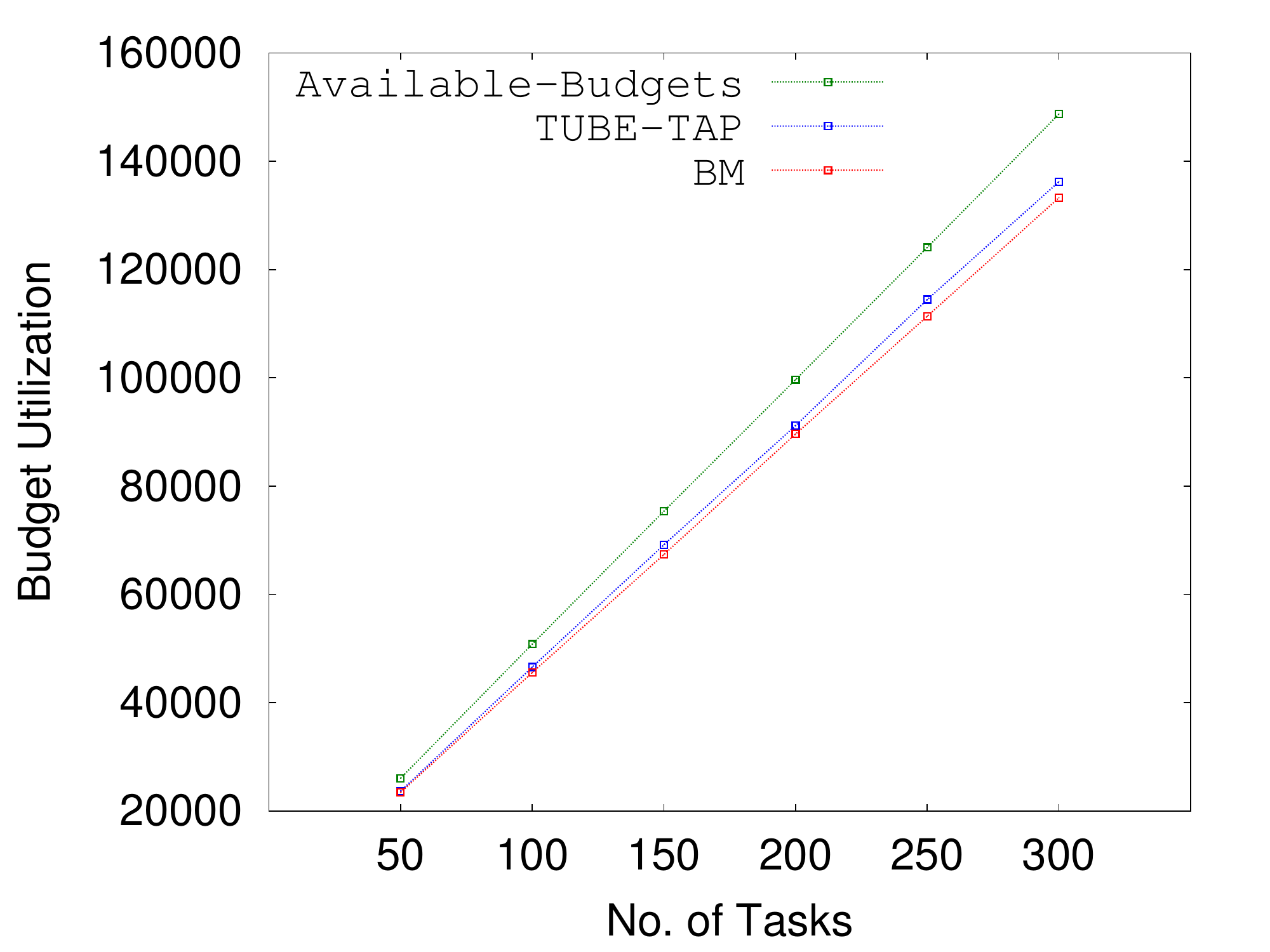}
                \subcaption{Budget utilization (UD)}
                \label{fig:sim2a}
        \end{subfigure}
        \caption{Comparison of Budget utilization in ND and UD cases}
\end{figure} 
\indent However, in case of BM each winner is paid a bit more than the bid value of preceding IoT device in the ordering. As the bid values of the IoT devices are sorted in increasing order, so the payment made to each winning IoT device in case of TUBE-TAP is more as compared to BM. Due to this reason, we can see that the budget utilization is higher in case of TUBE-TAP as compared to BM shown in Figure \ref{fig:sim1a} and Figure \ref{fig:sim2a}. Another important observation one can make from Figure \ref{fig:sim1a}, and Figure \ref{fig:sim2a} is that, both the mechanisms $i.e.$ TUBE-TAP and BM are \emph{budget feasible} that supports the claim made for TUBE-TAP in Lemma \ref{lemma:2}.\\
\indent Next comes the discussion on the behaviour of the mechanisms based on our second parameter. The sole purpose of considering this parameter is to judge the two mechanisms on the ground of \emph{truthfulness}. It is already pointed out that BM is vulnerable to manipulation $i.e.$ the IoT devices can gain by misreporting their privately held bid values in case of BM. During the simulation, in order to show the so called manipulative behaviour of BM we have varied the bid values of the subset of the IoT devices. More formally, we have considered that $15\%$ of the available IoT devices (in our case this is referred as \emph{small} variation) are increasing their bid value by 35\% of their true valuation. Similar is the case with \emph{medium} variation ($30\%$) and the \emph{large} variation ($40\%$). In the figures of simulation results, BM with \emph{small} variation, BM with \emph{medium} variation, and BM with \emph{large} variation is shown as BM-S-var, BM-M-var, and BM-L-var respectively.\\
\indent In Figure \ref{fig:sim1b} and Figure \ref{fig:sim2b} the comparison between the two mechanisms $i.e.$ TUBE-TAP and BM is done based on the \emph{utility of the IoT devices} parameter for ND and UD  cases respectively. 

\begin{figure}[H]
\begin{subfigure}[b]{0.50\textwidth}
                \centering
                \includegraphics[scale=0.60]{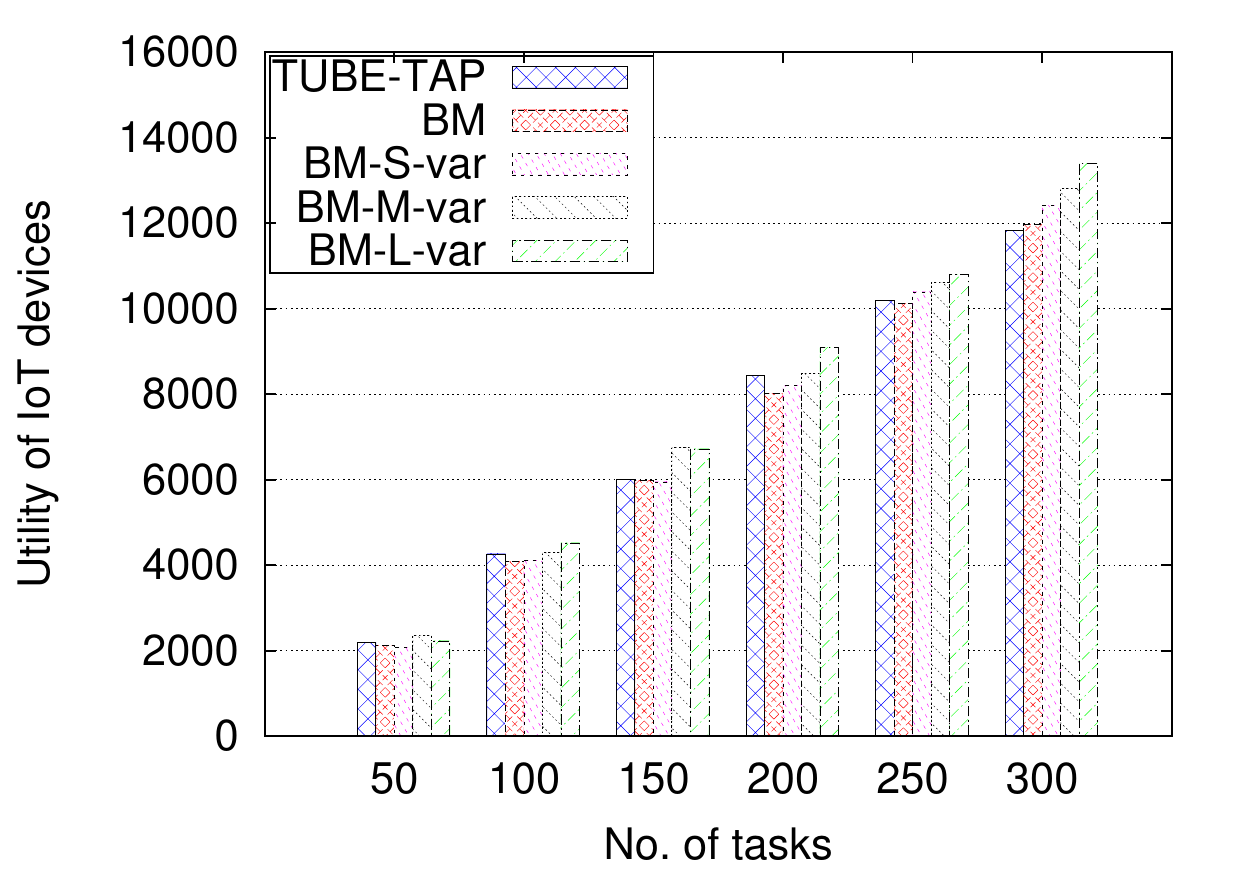}
                \subcaption{Utility of IoT devices (ND)}
                \label{fig:sim1b}
        \end{subfigure}%
        \begin{subfigure}[b]{0.50\textwidth}
                \centering
                \includegraphics[scale=0.60]{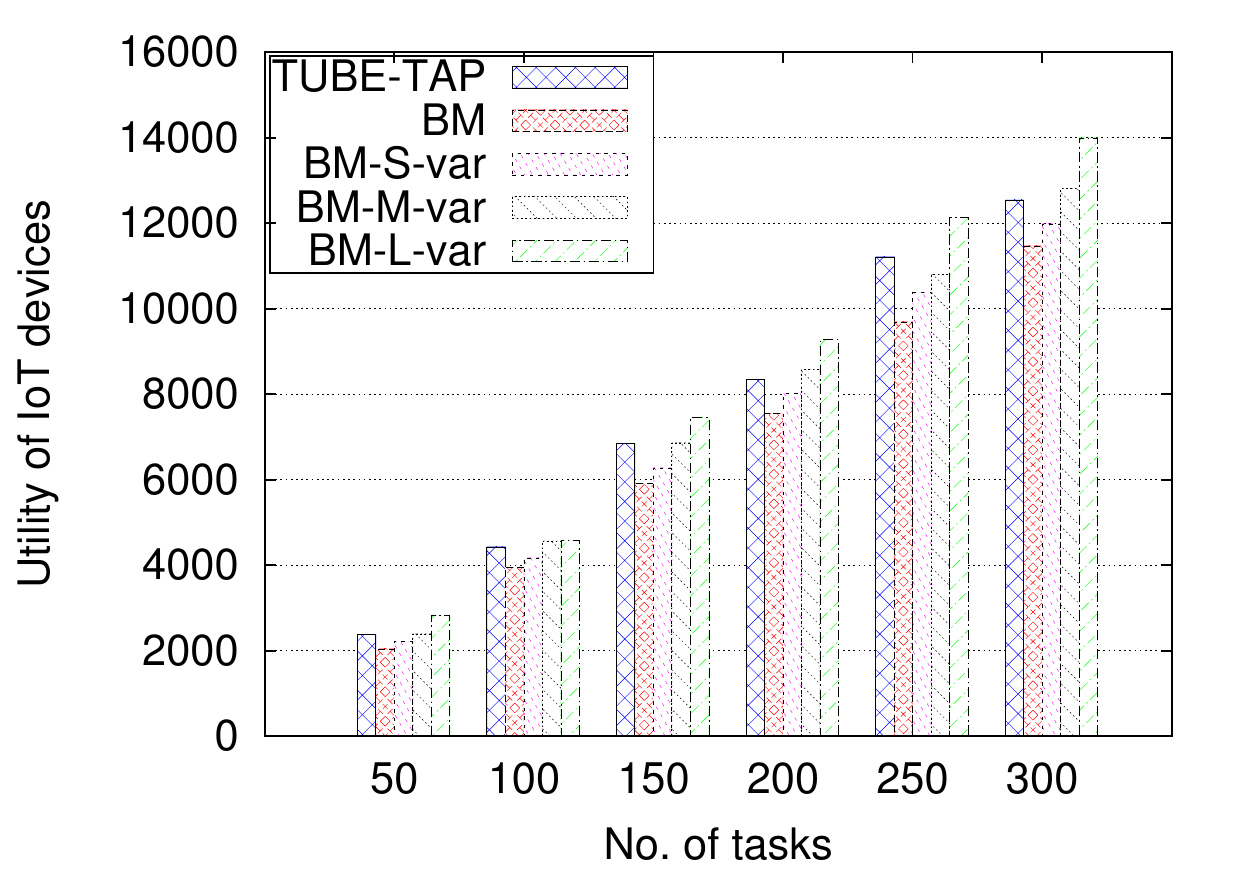}
                \subcaption{Utility of IoT devices (UD)}
                \label{fig:sim2b}
        \end{subfigure}
        \caption{Comparison of Utility of IoT devices in ND and UD cases}
\end{figure}

It can be seen that, most of the time the utility of IoT devices for TUBE-TAP is more as compared to the utility of IoT devices for BM in both ND and UD case. This very nature of TUBE-TAP is due to the reason that IoT devices are paid higher in case of TUBE-TAP as compared to BM that can be concluded from the results shown Figure \ref{fig:sim1a}, and Figure \ref{fig:sim2a}. Also, talking about the manipulative nature of the BM, it can be easily seen in Figure \ref{fig:sim1b}, and Figure \ref{fig:sim2b} that overall utility of the IoT devices gets increased by misreporting the bid values. The utility of IoT devices is higher in case of large variation than in case of medium variation than in case of small variation. Also, in some manipulative cases (mostly in \emph{large} variation) it could be seen that the utility achieved by the IoT devices in case of BM bypass even the utility gained by the IoT devices in case of TUBE-TAP. So, one can conclude that larger the number of IoT devices increasing their bid value by some amount (say 35\%) higher will be  the utility for the IoT devices. As the IoT devices are gaining by misreporting, so BM is \emph{non-truthful}. 

\section{Conclusion and Future Works}\label{se:cons}
In this paper, we have investigated a heterogeneous task assignment problem in IoT based crowdsourcing through the lens of mechanism design. We have designed a \emph{truthful} mechanism for the problem such that for each task the total payment made to the subset of IoT devices are within budget while achieving a threshold quality. 
In our future works, we can investigate the more realistic version of the problem by injecting the constraint that the tasks endowed with the task requesters are divisible in nature along with the several other additional constraints. We can think of designing a \emph{truthful budget feasible mechanism} for the more realistic version of the problem. 
\section*{Acknowledgements}
We would like to thanks the research students and faculty members of the Department of CSE, NIT Durgapur for their valuable suggestions during the course of this work. We would also like to thank Government of India, Ministry of Human Resource Development (MHRD) for the funds.

\bibliographystyle{plain}
\bibliography{phd}



\end{document}